\documentclass[11pt]{article}

\usepackage{fullpage}

\usepackage{graphicx}

\usepackage{amsmath}
\usepackage{amsthm}
\usepackage{amssymb}
\usepackage{natbib}
\usepackage{mathtools}
\usepackage[hidelinks]{hyperref}
\usepackage{bm}
\usepackage{enumitem}
\usepackage{ragged2e}

\usepackage[noabbrev]{cleveref}

\usepackage[dvipsnames]{xcolor}

\usepackage{pgfplots}
\usetikzlibrary{patterns}
\usepgfplotslibrary{fillbetween}
\usepackage{xfrac}
\usepackage{framed}
\usepackage[stable]{footmisc}
\pgfplotsset{compat=1.18}

\usepackage{amsmath,amsfonts,amssymb,amsthm,latexsym}
\usepackage{graphicx}
\usepackage{color}
\usepackage{bm}
\newtheorem{theorem}{Theorem}
\newtheorem{definition}{Definition}
\newtheorem{lemma}{Lemma}
\newtheorem{prop}{Proposition}
\newtheorem{fact}{Fact}
\newtheorem{corollary}{Corollary}

\newcommand{\agind}[1][i]{_{#1}}

\newcommand{\ironed}{\bar}
\newcommand{\constrained}{\hat}
\newcommand{\optconstrained}{\composed{\optimized}{\constrained}}
\newcommand{\optimized}[1]{#1\opt}
\newcommand{\differentiated}[1]{#1'}

\newcommand{\tagged}[2]{{#2}^{#1}}

\newcommand{\primedarg}[1]{#1\primed}
\newcommand{\noaccents}[1]{#1}
\newcommand{\composed}[3]{#1{#2{#3}}}

\newcommand{\newagentvar}[3][\noaccents]{%
\expandafter\newcommand\expandafter{\csname #2\endcsname}{#1{#3}}%
\expandafter\newcommand\expandafter{\csname #2s\endcsname}{#1{\boldsymbol{#3}}}%
\expandafter\newcommand\expandafter{\csname #2smi\endcsname}[1][i]{#1{\boldsymbol{#3}}_{-##1}}%
\expandafter\newcommand\expandafter{\csname #2i\endcsname}[1][i]{#1{#3}\agind[##1]}%
\expandafter\newcommand\expandafter{\csname #2ith\endcsname}[1][i]{#1{#3}_{(##1)}}%
}

\newcommand{\newitemvar}[3][\noaccents]{%
\expandafter\newcommand\expandafter{\csname #2\endcsname}{#1{#3}}%
\expandafter\newcommand\expandafter{\csname #2s\endcsname}{#1{\boldsymbol{#3}}}%
\expandafter\newcommand\expandafter{\csname #2smj\endcsname}[1][j]{#1{\boldsymbol{#3}}_{-##1}}%
\expandafter\newcommand\expandafter{\csname #2j\endcsname}[1][j]{#1{#3}_{##1}}%
\expandafter\newcommand\expandafter{\csname #2jth\endcsname}[1][j]{#1{#3}_{(##1)}}%
}

\newcommand{\forrezs}[1]{{#1}^{\rezs}}
\newagentvar[\forrezs]{imbal}{\phi}
\newagentvar[\forrezs]{cumimbal}{\Phi}

\newagentvar{alloc}{x}
\newagentvar{Dalloc}{\Delta x}

\newagentvar{falloc}{z}

\newagentvar{quant}{q}
\newagentvar[\constrained]{exquant}{\quant}
\newagentvar[\constrained]{critquant}{\quant}  

\newagentvar{Val}{\nu}
\newagentvar{toquant}{Q}

\newagentvar{qprice}{\price}
\newagentvar{qrev}{R}
\newagentvar[\ironed]{iqrev}{\qrev}

\newagentvar[\bar]{bstrat}{\strat}
\newagentvar[\bar]{bpay}{\pay}
\newagentvar[\bar]{bwal}{\wal}
\newagentvar[\bar]{bbid}{\bid}
\newagentvar{zpay}{\pay^{\rezs}}
\newagentvar{zwal}{\wal^{\rezs}}
\newagentvar{zdelt}{\Delta^{\rezs}}
\newagentvar{bomega}{\bm{\omega}}

\newagentvar{qalloc}{y}
\newagentvar{cumalloc}{Y}
\newagentvar[\constrained]{calloc}{\qalloc}
\newagentvar[\constrained]{cumcalloc}{\cumalloc}
\newagentvar[\ironed]{ialloc}{\qalloc}
\newagentvar[\ironed]{icumalloc}{\cumalloc}

\newcommand{\exposted}[1]{#1^{\text{\it EP}}}
\newagentvar[\composed{\exposted}{\constrained}]{excalloc}{\qalloc}
\newagentvar[\exposted]{exalloc}{\qalloc}
\newagentvar[\exposted]{extalloc}{\talloc}
\newagentvar[\exposted]{exfalloc}{\falloc}

\newagentvar{rev}{R}
\newagentvar[\differentiated]{marg}{\rev}
\newagentvar{rawrev}{P}
\newagentvar[\differentiated]{rawmarg}{\rawrev}

\newagentvar[\primedarg]{inducedrev}{\rev}

\newagentvar[\tilde]{pseudorev}{\rev}
\newagentvar[\tilde]{pseudorawrev}{\rawrev}

\newagentvar{typespace}{{\cal T}}
\newagentvar{typesubspace}{S}

\newagentvar{type}{t}
\newagentvar{othertype}{s}
\newagentvar{val}{v}
\newagentvar{outcome}{w}
\newagentvar{outcomespace}{{\cal W}}

\newcommand{\served}[1]{#1^1}
\newcommand{\nonserved}[1]{#1^0}
\newcommand{\alloced}[1]{#1^{\alloc}}
\newcommand{\allocedi}[1]{#1^{\alloci}}
\newagentvar[\alloced]{xoutcome}{\outcome}
\newagentvar[\allocedi]{xioutcome}{\outcome}
\newagentvar[\served]{soutcome}{\outcome}

\newagentvar[\nonserved]{nsoutcome}{\outcome}

\newagentvar{price}{p}
\newagentvar{talloc}{\alloc}
\newagentvar[\tilde]{eval}{\val}
\newagentvar{balloc}{y}
\newagentvar[\tilde]{dalloc}{y}
\newagentvar[\tilde]{eballoc}{y}
\newagentvar{ealloc}{y}
\newagentvar[\tilde]{bprice}{\price}

\newagentvar{act}{a}
\newagentvar{bidspace}{A}
\newagentvar{actspace}{A}

\newagentvar[\constrained]{critval}{\val}
\newagentvar[\constrained]{crittype}{\type}
\newagentvar[\constrained]{critvirt}{\virt}
\newagentvar[\constrained]{reserve}{\val} 
\newagentvar[\constrained]{bidreserve}{\bid} 
\newagentvar[\optconstrained]{monop}{\val}
\newagentvar[\constrained]{monot}{\type}
\newagentvar[\optimized]{monorev}{\rev}

\newagentvar{rcalloc}{y}
\newagentvar[\optimized]{optrcalloc}{\rcalloc}
\newagentvar{biddist}{G}

\newagentvar[\constrained]{critbid}{\bid}
\newagentvar[\constrained]{cbid}{B}
\newagentvar[\primedarg]{wbid}{\bid}
\newagentvar{gfunc}{\vartheta}
\newagentvar{block}{C}

\newagentvar{util}{u}

\newagentvar{strat}{\beta}
\newagentvar{bid}{b}

\newagentvar{pay}{\pi}
\newagentvar{rez}{\rho}
\newagentvar[\tilde]{trez}{\rez}

\newagentvar{virt}{\phi}
\newagentvar{cumvirt}{\Phi}
\newagentvar{qvirt}{\phi}
\newagentvar[\ironed]{ivirt}{\virt}
\newagentvar[\ironed]{icumvirt}{\cumvirt}

\newagentvar[\tagged{\text{SD}}]{sdvirt}{\virt}
\newagentvar[\composed{\tagged{\text{SD}}}{\ironed}]{sdivirt}{\virt}
\newagentvar[\tagged{\text{MD}}]{mdvirt}{\virt}
\newagentvar[\composed{\tagged{\text{MD}}}{\ironed}]{mdivirt}{\virt}

\newagentvar{dist}{F}
\newagentvar{dens}{f}
\newagentvar{hazard}{h}
\newagentvar{cumhazard}{H}

\newagentvar[\ironed]{iprice}{\price}  
\newagentvar[\ironed]{ival}{\val}  

\newagentvar{ints}{{\cal I}}

\newagentvar{wal}{w}
\newagentvar{Util}{U}

\newitemvar{pos}{j}
\newitemvar{weight}{w}
\newitemvar[\differentiated]{mweight}{\weight}
\newitemvar{udtype}{\type}
\newitemvar[\constrained]{udcrittype}{\type}
\newitemvar{udalloc}{\alloc}
\newitemvar{udprice}{\price}
\newitemvar[\constrained]{udcalloc}{\qalloc}
\newitemvar[\constrained]{udcumcalloc}{\cumalloc}

\newagentvar{mech}{{\cal M}}
\newagentvar[\skew{5}{\hat}]{cmech}{{\cal M}}
\newagentvar{alg}{{\cal A}}

\newcommand{\game}{\mathcal{G}}

\DeclareMathOperator{\OPT}{OPT}

\newcommand{\algo}{A}

\newcommand{\reals}{{\mathbb R}}

\newagentvar{trans}{\sigma}

\newagentvar{demandset}{S}

\newcommand{\opt}{^{\star}}
\newcommand{\primed}{^\dagger}

\newagentvar{distout}{w}
\newagentvar{idistout}{\bar{w}}

\newagentvar{gap}{\delta}

\newcommand{\benchmark}{B}
\newcommand{\bmk}{\benchmark}

\newcommand{\expecta}{\mathbf{E}}

\newcommand{\scF}{\mathcal{F}}

\newagentvar{blend}{\delta}
\newagentvar{abstrd}{\gamma}
\newagentvar{abstrdbar}{\bar{\gamma}}
\newagentvar{abstrden}{y}

\newagentvar{minc}{C}
\newagentvar{Spaceom}{\Omega}
\newagentvar{spaceom}{\omega}
\newagentvar{spaceombar}{\bar{\omega}}
\newcommand{\despace}{\Spaceomi[1]}
\newcommand{\deact}{\spaceomi[1]}
\newcommand{\deactbar}{\spaceombari[1]}

\newcommand{\decost}{\minci[1]}
\newcommand{\stspace}{\Spaceomi[2]}
\newcommand{\stact}{\spaceomi[2]}

\newcommand{\simplex}{\Delta}
\newcommand{\stcost}{\minci[2]}

\begin{document}



\title{Extension of Minimax for Algorithmic Lower Bounds}

\newcommand{\email}[1]{\href{mailto:#1}{#1}}

\author{Jason Hartline\thanks{Northwestern U., Evanston, IL.  Work done in part while supported by NSF CCF 1618502.\newline Email: \email{hartline@northwestern.edu}} \and Aleck Johnsen\thanks{Northwestern U., Evanston, IL.  Work done in part while supported by NSF CCF 1618502.\newline Email: \email{aleckjohnsen@u.northwestern.edu}}}



\maketitle
\begin{abstract}
This paper considers the use of Yao's Minimax Theorem in robust algorithm design, e.g., for online algorithms, where the algorithm designer aims to minimize the ratio of the algorithm's performance to the optimal performance.  When applying Minimax, the objective becomes the expectation of these ratios.  The main result of the paper is that it is equally valid for the objective (after applying Minimax) to be the ratio of the expectations.
\end{abstract}

\section{Overview}
\label{s:intro}

This paper presents an extension of Yao's Minimax Theorem \citep{yao-77} that is a pertinent and potent analytical tool within computer science theory.  The main application of this paper relates (quite generally) to the measurement of ``robustness" within robust algorithm design, e.g., for the standard model of competitive analysis in online algorithms.

The general model for Minimax is worst-case minimization which is a two-stage optimization.  A fixed cost (objective) function $\minc$ takes two kinds of inputs: (1) designer variables (i.e., cost parameters) and (2) descriptions of state.  A designer assigns variables and evaluates cost subject to a requirement that state is deterministic to maximize cost (given the designer's fixed assignments).  Yao's Minimax Theorem (see \Cref{thm:yaominmax}) says that a lower bound on the designer's optimal cost can be analyzed by choosing {\em any fixed randomization} over state variables and letting the designer best respond to minimize the cost function {\em in expectation} over the mixed state.  Minimax is a powerful technique because for many optimization problems, there exists a choice of mixed-state that leads to description of the tight lower bound.

A challenge arises within application of Minimax for ratio cost functions, i.e., costs $\minc (\cdot)\vcentcolon= \sfrac{\bmk(\cdot)}{\algo(\cdot)}$.  Technically, the problem is that central to the lower bound analysis is an objective function taking expectation of cost over states in a way that the ratio is {\em inside} the expectation: $\expecta\left[\minc(\cdot)\right]=\expecta\left[\sfrac{\bmk(\cdot)}{\algo(\cdot)} \right]$, which we call {\em expectation-over-ratios} (EOR).  Intuitively, the challenge is that weighted sums of ratios are difficult to manipulate and quantify algebraically (unless every denominator is the same, or every ratio is the same).  In particular, we generally can not analyze the expected cost by analyzing the functions $\bmk$ and $\algo$ independently.

By contrast, consider the similar -- but algebraically quite distinct -- quantity $\sfrac{\expecta\left[\bmk(\cdot)\right]}{\expecta\left[\algo(\cdot)\right]}$ (which uses the same distribution weights), which we call {\em ratio-of-expectations} (ROE).  This quantity is easy to manipulate (e.g., by multiplying through an equation in which it appears by $\expecta\left[ \algo(\cdot)\right]$).  And, this quantity is relatively easy to calculate, due to the independence of the expectations applied respectively to $\bmk$ and to $\algo$, with only a single ratio to consider afterwards.

{\em The main result of this paper shows that we can write Minimax-style lower bounds as ratio-of-expectations rather than the prevailing form of expectation-over-ratios} (see \Cref{lem:mixedbenchmark}).

To prove the main result, the key technique is an initial step to change the objective function to expectation-over-ratios, {\em before applying Minimax} (see \Cref{lem:mixedstate}).

As stated above, in fact the original expectations-over-ratios structure can be easy to manipulate {\em if every ratio is the same}.  To explicitly illustrate the analytical strength of this constant-ratio property, we re-prove a version of the main result: in a weaker setting, i.e., stronger assumptions; albeit, with a stronger result: programs using the original ratio objective, EOR, and ROE {\em all have equal optimal value} (see \Cref{a:mainresultminmax}).

\paragraph{Illustration of the Main Result}

To illustrate both the key challenge of ratio cost functions as transformed to EOR under standard Minimax, and the simplification gained by writing lower bounds rather as ROE, consider broadly the field of online algorithms and its canonical framework to measure robustness: competitive analysis.

For an online algorithms problem, an algorithm designer will observe a sequence $\spaceom\in\Spaceom$ of state-elements.  For each element of the sequence, the algorithm must commit to an action with knowledge of past state but before seeing the current or future states. 
 Finally, given the algorithm's actions and the realized state-sequence as inputs, an algorithm $\algo$ is awarded a positive utility $\algo(\spaceom)$.  (Note, maximizing this utility -- in its independent form -- is not the ultimate objective function for the designer; see below.)

Competitive analysis defines a robust objective for algorithm design.  For each sequence $\spaceom$ of states, it calculates the utility of the optimal algorithm that knows the full sequence in advance, a.k.a., the {\em offline optimal} algorithm $\OPT_{\spaceom}$ which has positive utility $\OPT_{\spaceom}(\spaceom)$.  For each sequence $\spaceom$, it sets a benchmark $\bmk(\spaceom)=\OPT_{\spaceom}(\spaceom)$ and critically, {\em defines a cost function for use within worst-case minimization} by $\minc(\spaceom)\vcentcolon=\sfrac{\OPT_{\spaceom}(\spaceom)}{\algo(\spaceom)}$.

To summarize competitive analysis: a pre-committed online algorithm is designed with the objective to minimize the cost $\sfrac{\OPT_{\spaceom}(\spaceom)}{\algo(\spaceom)}$ for the $\spaceom\in\Spaceom$ inducing the worst (largest-cost) ratio, i.e., it identifies an algorithm's worst-case multiplicative-approximation of the offline-optimal benchmark.

Applying standard Minimax, the objective function (used for lower bounds) has the form $\expecta\left[\sfrac{\OPT_{\spaceom}(\spaceom)}{\algo(\spaceom)}\right]$.  In this case, analysis depends pointwise on the utilities $\OPT_{\spaceom}(\spaceom)$.  The best-response algorithm must approximate these $\spaceom$-state benchmarks individually, showing preference towards minimizing the ratios of certain states based on both the absolute value of the induced benchmark and the density assigned to the state by the distribution.  Frequently, the calculation of a best-response algorithm's per-state utility tradeoffs is prohibitively complex.

By contrast, consider an objective function (for lower bounds) with the form $\sfrac{\expecta\left[\OPT_{\spaceom}(\spaceom)\right]}{\expecta\left[\algo(\spaceom)\right]}$ (the validity of which we will prove in our main result).  The reduction of complexity is apparent.  Critically, the dependence on ratios between each state's optimal-algorithm utility and designer-algorithm utility is broken.  Instead, there is a single benchmark-quantity in the numerator which is the expectation over the individual-state benchmarks.  And correspondingly, there is a single quantity in the denominator which is the algorithm's expected utility over states (for which the fixed-algorithm does not know the realized state, distinguishing this quantity from the numerator).  In conclusion for a given mixture over sequences $\spaceom$: the numerator is a fully-determined constant; and the optimal algorithm is simply, independently the best response to the mixture over states.

\paragraph{Related Work} \cite{ABB-23} give a model for robust algorithm design that unifies the regret and ratio objectives, using themes that overlap this paper.  Their general objective considers for a given mixture over states, the additive difference between two weighted terms: expectation-over-benchmarks and expectation-over-algorithm-utilities.  These objective functions -- in conjunction with the respective, realized optimal values of their resulting Minimax programs -- can be algebraically rearranged to give a generalized model of our ratio-of-expectations objective (where the generalization is due to their weight parameter).

\section{Formal Minimax Preliminaries}
\label{s:setup}

Denote a Hilbert space by $\Spaceom$.\footnote{\label{foot:hilbert} An informal definition of a Hilbert space is: it is the generalization of the structure of a discretely-dimensioned space to arbitrary discrete and/or continuous dimensions, including generalization of discrete probability measures; e.g., the structure of a linear program in discrete variables/constraints theoretically generalizes to be continuous.}  Given arbitrary $\Spaceom$, denote the set of all possible distributions by $\simplex(\Spaceom)$ -- i.e., the {\em probability simplex}.  Denote a distribution over elements $\spaceom\in\Spaceom$ by $\abstrd\in\simplex(\Spaceom)$.\footnote{\label{foot:pointmass} In this paper, we allow distributions to include point masses with strictly positive measure as Dirac functions.}



Worst-case minimization of a cost function $\minc$ is parameterized by a partition of the domain-variables of $\minc$ into two sets (i.e. spaces, which have known support): (1) optimization variables chosen by the {\em designer}; and (2) dimensions of state over which the worst-case value of $\minc$ is analyzed. 
 Respectively, denote the design space by $\despace$ and the state space by $\stspace$.  Thus, the cost function is formally defined by $\minc:\despace \times \stspace \rightarrow \reals_+$ (with an assumption that costs are non-negative).\footnote{\label{foot:internalrand} There is a generalization of this cost function which allows the designer to use {\em internal randomization} over its original space $\despace$.  This generalization is material for some applications but is not necessary for the presentation of this paper.}
 
 
 

Worst-case minimization is represented mathematically as a successive $\min$-$\max$ analysis and is thus called {\em Minimax}:

\begin{equation}
\label{eqn:minmaxpure}
\min_{\deact\in\despace}~\max_{\stact\in\stspace}\minc(\deact,\stact)
\end{equation}

\noindent Most generally -- to account for design/state spaces that are not compact which may require analysis in-the-limit -- Minimax is represented as $\inf$-$\sup$:
\begin{equation}
\label{eqn:infsuppure}
\inf_{\deact\in\despace}~\sup_{\stact\in\stspace}\minc(\deact,\stact)
\end{equation}

\subsection{Theoretical Lower Bounds from Minimax}
\label{s:yaoandid}
\label{s:pitechnique}
\label{s:adversary}

The formulation of equation~\eqref{eqn:infsuppure} illustrates the natural interpretation of Minimax as a two-player, zero-sum game between the designer whose objective is to minimize (in worst-case) $\decost(\cdot)=\minc(\deact,\stact)$ and an {\em adversary} whose objective is to minimize $\stcost(\cdot)=-\minc(\deact,\stact)$.\footnote{\label{foot:zerosum} {\em Zero-sum} games satisfy: $\forall~ (\deact,\stact)$, sum utility of both players is identically $\decost(\deact,\stact)+\stcost(\deact,\stact)=0$.}  Moreover, {\em first} the designer explicitly commits to an {\em action} $\deact\in\despace$, and then {\em second} the adversary best-responds with an action $\stact\in\stspace$ that is worst-case (for the designer).

Yao's Minimax Principle (\Cref{thm:yaominmax}) states that the order of Minimax actions may be reversed to $\sup$-$\inf$ with two stipulations: (a) the adversary (now acting first) may randomize its action; and (b) in the case that originally $\inf$-$\sup$ is strictly necessary (because $\min$-$\max$ is not well-defined), the optimal value of $\sup$-$\inf$ may be a non-tight lower bound.

\begin{theorem}[\citealp{yao-77}]
\label{thm:yaominmax}
{\em [Yao's Minimax Principle]}  Given a $2$-player zero-sum game $\game$ in which sequentially player $1$ chooses an action $\deact\in\despace$, then player $2$ chooses action $\stact\in\stspace$.  Given the players are cost minimizers and the cost functions on pure actions are (any real-valued function) $\minci[1](\deact,\stact)\geq 0$ and $\minci[2]=-\minci[1]$.  Then the {\em value} of game $\game$ (the left-hand side) satisfies:
\begin{align}
\label{eqn:yaochain}
    \inf_{\deact\in\despace}~\sup_{\stact\in\stspace}\minci[1](\deact,\stact)&\geq \sup_{\abstrdi[2]\in\simplex(\stspace)}~\inf_{\deact\in\despace}\expecta_{\stact\sim\abstrdi[2]}\left[\minci[1](\deact,\stact)\right]\\
    \nonumber
    &\geq \inf_{\deact\in\despace}\expecta_{\stact\sim\abstrdbari[2]}\left[\minci[1](\deact,\abstrdbari[2])\right],~~\forall~\text{\em fixed}~\abstrdbari[2]\in\simplex(\stspace)
\end{align}
\end{theorem}

\paragraph{Lower Bounds for Ratio Objectives}

Consider the special case of the game description in \Cref{thm:yaominmax} in which the cost function is a ratio.  For a fixed input $(\deact,\stact)$ and two functions $\algo:\despace \times \stspace \rightarrow \reals_+$ and $\bmk:\despace \times \stspace \rightarrow \reals_{0+}$ define:\footnote{\label{foot:denomgzero} {Of course this definition can always be trivially assigned by setting $\algo(\deact,\stact)=1$, but our interest is in cases where there is an economically-meaningful multiplicative separation of $\minc$, e.g., robust algorithm design.
Moreover, we will restrict the functional definition of $\algo(\cdot)$ to require that: there exists $\deact\in\despace$ such that $\algo(\deact,\stact)>0$ for all $\stact$.  Thus, cost is guaranteed to be finite, which is consistent with our original definition $\minc:\despace \times \stspace \rightarrow \reals_+$.}}
\begin{equation}
\label{eqn:costratio}
    \minc(\spaceomi[1],\spaceomi[2])\vcentcolon= \frac{\bmk(\spaceomi[1],\spaceomi[2])}{\algo(\spaceomi[1],\spaceomi[2])}<\infty
\end{equation}

\noindent We exhibit Yao's Minimax Principle (\Cref{thm:yaominmax}) applied to a ratio cost function, including explicit substitution of the objectives in equation~\eqref{eqn:yaochain} to show it as an \textbf{expectation-over-ratios}:

\begin{align}
    \label{eqn:yaoeor}
\inf_{\deact\in\despace}~\sup_{\stact\in\stspace}\left[\frac{\bmk(\spaceomi[1],\spaceomi[2])}{\algo(\spaceomi[1],\spaceomi[2])}\right]
&\geq \sup_{\abstrdi[2]\in\simplex(\stspace)}~\inf_{\deact\in\despace}\left[\expecta_{\stact\sim\abstrdi[2]}\left[\frac{\bmk(\deact,\stact)}{\algo(\deact,\stact)}\right]\right]\\
\nonumber
& \geq \inf_{\deact\in\despace}\left[\expecta_{\stact\sim\abstrdbari[2]}\left[\frac{\bmk(\deact,\stact)}{\algo(\deact,\stact)}\right]\right],~~\forall~\text{fixed}~\abstrdbari[2]\in\simplex(\stspace)
\end{align}

\noindent Obtaining lower bounds from EOR -- tight or otherwise -- is frequently challenging because it is difficult to algebraically simplify the expectation.  See \Cref{s:intro} for a brief discussion of this observation.

\subsection{Replacing Expectation-over-Ratios with Ratio-of-Expectations}
\label{s:maintheorem}

The goal of this section is to state and prove the main result of this paper: within the central equation of Yao's Minimax Principle applied to a ratio cost objective, we may replace expectation-over-ratios with ratio-of-expectations.  We give this result below in \Cref{lem:mixedbenchmark}.

First, we identify the technical innovation that will be employed in the proof of \Cref{lem:mixedbenchmark}, which we state and prove as \Cref{lem:mixedstate}.

\begin{lemma}[The Randomized-state Relaxation Lemma]
\label{lem:mixedstate}
Consider a worst-case minimization problem setting with design space $\despace$ and state space $\stspace$.  Without loss of generality for the value of designer optimization, we can simultaneously: relax the state space to allow any mixture over states, and transition EOR to ROE (while otherwise maintaining the original $\inf$-$\sup$ framework). 
 I.e,, the following $\inf$-$\sup$ programs necessarily have the same value:
\begin{equation}
\label{eqn:mixedstatelemma}
\inf_{\deact\in\despace}~\sup_{\stact\in\stspace}\left[\frac{\bmk(\spaceomi[1],\spaceomi[2])}{\algo(\spaceomi[1],\spaceomi[2])}\right]
=\inf_{\deact\in\despace}~\sup_{\abstrdi[2]\in\simplex(\stspace)} \left[\frac{\expecta_{\stact\sim\abstrdi[2]}\left[\bmk(\deact,\stact)\right]}{\expecta_{\stact\sim\abstrdi[2]}\left[\algo(\deact,\stact)\right]}\right]
\end{equation}
\end{lemma}

\begin{proof}
\noindent The right-hand program makes two simultaneous modifications to the original cost-ratio Minimax program on the left-hand side: (1) it relaxes the adversary's action space to allow mixed states, i.e., to allow $\abstrdi[2]\in\simplex(\stspace)$; and (2) it changes the objective from expectation-over-ratios (EOR) to ratio-of-expectations (ROE).  We prove the lemma by showing equality of the values of the inner maximization ($\sup$) programs (which is sufficient, because then the value of the adversary's response is equal for any (previously fixed) choice of $\deactbar\in\despace$.

Every $\stact\in\stspace$ on the left-hand side is available on the right-hand side as a pointmass distribution.  This proves:

\begin{equation*}
\sup_{\stact\in\stspace}\left[\frac{\bmk(\deactbar,\stact)}{\algo(\deactbar,\stact)}\right]
\leq\sup_{\abstrdi[2]\in\simplex(\stspace)} \left[\frac{\expecta_{\stact\sim\abstrdi[2]}\left[\bmk(\deactbar,\stact)\right]}{\expecta_{\stact\sim\abstrdi[2]}\left[\algo(\deactbar,\stact)\right]}\right],~~\forall~\text{fixed}~\deactbar\in\despace
\end{equation*}

\noindent \Cref{lemma:roedominance} (in \Cref{a:monopropsoffractions}) shows that for every distribution $\abstrdbari[2]\in\simplex(\stspace)$, there exists an element of $\abstrdbari[2]$ for which the element's individual ratio (i.e., a lower bound on the value of the right-hand side) is at least the value of the left-hand side (with its ratio-of-expectations cost function).  This proves:

\begin{equation*}
\sup_{\stact\in\stspace}\left[\frac{\bmk(\deactbar,\stact)}{\algo(\deactbar,\stact)}\right]
\geq\sup_{\abstrdi[2]\in\simplex(\stspace)} \left[\frac{\expecta_{\stact\sim\abstrdi[2]}\left[\bmk(\deactbar,\stact)\right]}{\expecta_{\stact\sim\abstrdi[2]}\left[\algo(\deactbar,\stact)\right]}\right],~~\forall~\text{fixed}~\deactbar\in\despace\qedhere
\end{equation*}
\end{proof}

\noindent Intuitively, the key reason that \Cref{lem:mixedstate} holds with equality is: the adversary already acts second such that relaxing to allow distributions over elements of $\stspace$ gives no extra advantage to the adversary.  Explained differently: worst-case mixtures can not be ``more powerful" than the worst element within the mixture.  The technical contribution of \Cref{lem:mixedstate} towards our goal is quite strong: note that we successfully manipulated -- {\em in advance of applying Minimax} -- the objective function to exhibit the structure of \textbf{ratio-of-expectations}.

We give our main result next as \Cref{lem:mixedbenchmark}.  However as a summary of the proof, the statement now follows immediately from standard application of Minimax to the program resulting from \Cref{lem:mixedstate}.

\begin{lemma}[The Ratio-of-Expectations Minimax Lemma]
\label{lem:mixedbenchmark}
Consider a worst-case minimization problem setting with design space $\despace$ and state space $\stspace$.  Let $\abstrdbari[2]$ be any fixed mixture over states. Then

\begin{align}
\label{eqn:mixedbenchmarklemma}
\inf_{\deact\in\despace}~\sup_{\stact\in\stspace}\left[\frac{\bmk(\spaceomi[1],\spaceomi[2])}{\algo(\spaceomi[1],\spaceomi[2])}\right]
&\geq \sup_{\abstrdi[2]\in\simplex(\stspace)}~\inf_{\deact\in\despace}\left[\frac{\expecta_{\stact\sim\abstrdi[2]}\left[\bmk(\deact,\stact)\right]}{\expecta_{\stact\sim\abstrdi[2]}\left[\algo(\deact,\stact)\right]}\right]\\
\nonumber
& \geq \inf_{\deact\in\despace}\left[\frac{\expecta_{\stact\sim\abstrdbari[2]}\left[\bmk(\deact,\stact)\right]}{\expecta_{\stact\sim\abstrdbari[2]}\left[\algo(\deact,\stact)\right]}\right],~~\forall~\text{\em fixed}~\abstrdbari[2]\in\simplex(\stspace)
\end{align}
\end{lemma}

\begin{proof}
\noindent We start with the ratio-cost worst-case minimization problem.  Explanations for each step of this sequence are given following.

\begin{align*}
\inf_{\deact\in\despace}~\sup_{\stact\in\stspace}\left[\frac{\bmk(\spaceomi[1],\spaceomi[2])}{\algo(\spaceomi[1],\spaceomi[2])}\right]
&=\inf_{\deact\in\despace}~\sup_{\abstrdi[2]\in\simplex(\stspace)} \left[\frac{\expecta_{\stact\sim\abstrdi[2]}\left[\bmk(\deact,\stact)\right]}{\expecta_{\stact\sim\abstrdi[2]}\left[\algo(\deact,\stact)\right]}\right]\\
&\geq \sup_{\abstrd\in\simplex(\simplex(\stspace))}~\inf_{\deact\in\despace}\left[\expecta_{\abstrdi[2]\sim \abstrd}\left[\frac{\expecta_{\stact\sim\abstrdi[2]}\left[\bmk(\deact,\stact)\right]}{\expecta_{\stact\sim\abstrdi[2]}\left[\algo(\deact,\stact)\right]}\right]\right]\\
&\geq \sup_{\abstrdi[2]\in\simplex(\stspace)}~\inf_{\deact\in\despace}\left[\frac{\expecta_{\stact\sim\abstrdi[2]}\left[\bmk(\deact,\stact)\right]}{\expecta_{\stact\sim\abstrdi[2]}\left[\algo(\deact,\stact)\right]}\right]\\
&\geq \inf_{\deact\in\despace}\left[\frac{\expecta_{\stact\sim\abstrdbari[2]}\left[\bmk(\deact,\stact)\right]}{\expecta_{\stact\sim\abstrdbari[2]}\left[\algo(\deact,\stact)\right]}\right],~~\forall~\text{fixed}~\abstrdbari[2]\in\simplex(\stspace)
\end{align*}

\begin{itemize}
    \item The first line holds as the exact statement of \Cref{lem:mixedstate}.
    \item The second line applies Yao's Minimax Principle (\Cref{thm:yaominmax}).  Note, the adversary's choice of actions $\abstrd\in\simplex(\simplex(\stspace))$ represents the exact transformation using Minimax: the adversary now acts first and plays a distribution over actions in $\simplex(\stspace)$ from the initial $\inf$-$\sup$ side.
    \item The third line holds by restricting the choice of $\abstrd\in\simplex(\simplex(\stspace))$ to its sub-class of pointmass distributions on each fixed $\abstrdbari[2]\in\simplex(\stspace)$ (and then deleting the redundant-structure to describe a fixed $\abstrdbari[2]$ as ``a pointmass on $\abstrdbari[2]$").  Restricting this choice of $\abstrd$ can only impair the new $\sup$-objective in the third line.
    \item The last line holds because fixing an argument of the outer $\sup$-program can only possibly impair its maximization objective (in this case by fixing $\abstrdi[2]=\abstrdbari[2]$ for any $\abstrdbari[2]\in\simplex(\stspace)$.\qedhere
\end{itemize}
\end{proof}

\noindent Note, the statements of \Cref{lem:mixedstate} and \Cref{lem:mixedbenchmark} hold for the most-general Minimax setting using $\inf$-$\sup$ (which does not require, e.g., that spaces be compact).  The weaker $\min$-$\max$ setting -- ``weaker" in the sense that it imposes stronger requirements to guarantee existence of the program's value -- remains interesting, given its simpler description and analytical requirements (e.g., it obviates in-the-limit).

To invoke the $\min$-$\max$ setting, we need only restrict attention to cost minimization problems {\em for which this value exists}, and frequently, this can be proved by guess-and-check (i.e., assume it is true, then solve for the value which existence proves the assumption is correct).

For the reduced $\min$-$\max$ setting, we re-prove the ROE lower-bounds main-result of this paper (i.e., analogous to \Cref{lem:mixedbenchmark}) because we believe its simpler proof is informative.  We give this statement and result in \Cref{a:mainresultminmax}.

\bibliographystyle{apalike}
\bibliography{bib}

\begin{appendix}

\section{Re-proof of the Main Result in a Weaker Setting}
\label{a:mainresultminmax}

\begin{lemma}[The Weak Ratio-of-Expectations Minimax Lemma]
\label{lem:mixedbenchmarkweak}
Consider a worst-case minimization problem setting with design space $\despace$ and state space $\stspace$.  Let $\abstrdbari[2]$ be any fixed mixture over states. Further, assume: {\em the problem setting is known to have a realized value under $\min-\max$ analysis} (i.e., it does not require generalization to $\inf$-$\sup$).  Then:

\begin{align*}
\label{eqn:mixedbenchmarklemmaminmax}
\min_{\deact\in\despace}~\max_{\stact\in\stspace}\left[\frac{\bmk(\spaceomi[1],\spaceomi[2])}{\algo(\spaceomi[1],\spaceomi[2])}\right]
&=\min_{\deact\in\despace}~\max_{\abstrdi[2]\in\simplex(\stspace)}\left[\expecta_{\stact\sim\abstrdi[2]}\left[\frac{\bmk(\spaceomi[1],\spaceomi[2])}{\algo(\spaceomi[1],\spaceomi[2])}\right]\right]\\
&=\min_{\deact\in\despace}~\max_{\abstrdi[2]\in\simplex(\stspace)}\left[\frac{\expecta_{\stact\sim\abstrdi[2]}\left[\bmk(\deact,\stact)\right]}{\expecta_{\stact\sim\abstrdi[2]}\left[\algo(\deact,\stact)\right]}\right]\\
&\geq \max_{\abstrd\in\simplex(\simplex(\stspace))}~\min_{\deact\in\despace}\left[\expecta_{\abstrdi[2]\sim\abstrd} \left[\frac{\expecta_{\stact\sim\abstrdi[2]}\left[\bmk(\deact,\stact)\right]}{\expecta_{\stact\sim\abstrdi[2]}\left[\algo(\deact,\stact)\right]}\right]\right]\\
&\geq \max_{\abstrdi[2]\in\simplex(\stspace)}~\min_{\deact\in\despace}\left[\frac{\expecta_{\stact\sim\abstrdi[2]}\left[\bmk(\deact,\stact)\right]}{\expecta_{\stact\sim\abstrdi[2]}\left[\algo(\deact,\stact)\right]}\right]\\
&\geq \min_{\deact\in\despace}\left[\frac{\expecta_{\stact\sim\abstrdbari[2]}\left[\bmk(\deact,\stact)\right]}{\expecta_{\stact\sim\abstrdbari[2]}\left[\algo(\deact,\stact)\right]}\right],~~\forall~\text{\em fixed}~\abstrdbari[2]\in\simplex(\stspace)
\end{align*}

\noindent Explicitly: the original program's ratio objective, the EOR objective, and the ROE objective {\em all have equal value} (in the first two lines).
\end{lemma}

\begin{proof}
We start with the ratio-cost worst-case minimization problem.  Explanations for each step of this sequence are given following.  (Note, the outline of the proof is similar to the proofs of \Cref{lem:mixedstate} and \Cref{lem:mixedbenchmark}; where explanations are the same as in those proofs, we may reduce detail here.)

\begin{itemize}
    \item The first line above holds with equality because the adversary already ``acts" second such that relaxing to allow distributions over elements of $\scF$ gives no extra advantage to the adversary (cf., the proof of \Cref{lem:mixedbenchmark} for the general case).\label{page:blendsconvexfractions}
    \item The second line changes expecation-over-ratio (EOR) to ratio-over-expectations (ROE), still within the original $\min$-$\max$ structure.  It holds with equality because if we do describe the adversary's optimal action with a proper mixture (not a point mass), then to satisfy optimality, the adversary must only put positive measure on states that all achieve the same optimal ratio, such that EOR is equal to ROE.\footnote{\label{foot:constratioequals} When all ratios are the same -- albeit with arbitrary scale -- then EOR equal to ROE holds as an identity; e.g., $\frac{1}{2}\cdot\frac{4}{10}+\frac{1}{2}\cdot\frac{2}{5}=\frac{2}{5} = \frac{\sfrac{1}{2}\cdot4+\sfrac{1}{2}\cdot2}{\sfrac{1}{2}\cdot10+\sfrac{1}{2}\cdot5}$.}

    (Note: it is this second step that does not go through for the general case in \Cref{lem:mixedstate} and \Cref{lem:mixedbenchmark}, requiring it to identify and apply its distinct technical manipulation for its proof.)
    \item The third, fourth, and fifth lines have the same justification as the final lines of \Cref{lem:mixedbenchmark}:
    \begin{itemize}
        \item the third lines Yao's Minimax Principle (\Cref{thm:yaominmax});
        \item the fourth lines restricts the choice of $\abstrd$ to pointmass distributions (for each fixed $\abstrdbari[2]$);
        \item and the last line fixes any given $\abstrdbari[2]\in\simplex(\stspace)$ (from the class of pointmass distributions in the previous point).\qedhere
    \end{itemize}
\end{itemize}
\end{proof}


\section{Inequality Properties of Re-weighted Fractions}
\label{a:monopropsoffractions}

\noindent \Cref{lemma:roedominance} supports \Cref{lem:mixedstate} in \Cref{s:maintheorem}.  It states that for a ROE objective like the left-hand side of equation~\eqref{eqn:mixedstatelemma} (which structure is our intended simplification), there must exist a single element of any state-space mixture that achieves at least the (inner-maximization) value of the overall ROE.  This statement is a corollary of a standard statement from the probabilistic method:

\begin{fact}
    \label{fact:elementversusexpect}
    Given an arbitrary distribution $\xi$ over $\reals$.  There exists $x$ in the positive-density support $\xi$ such that $x\geq\expecta_{\hat{x}\sim\xi}$.
\end{fact}

\begin{lemma}
\label{lemma:roedominance}
Consider a domain $\Omega$ and two positive functions $s:\Omega\rightarrow\reals_+$ and $t:\Omega\rightarrow\reals_+$.  For every distribution $\abstrd\in\simplex(\Omega)$, there exists $\omega_+$ in the positive-density support of $\abstrd$ for which:
\begin{equation}
    \label{eqn:roedominance}
    \frac{s(\omega_+)}{t(\omega_+)} \geq \frac{\expecta_{\omega\sim\abstrd}\left[ s(\omega)\right]}{\expecta_{\omega\sim\abstrd}\left[ t(\omega)\right]}
\end{equation}
\end{lemma}
\begin{proof}
Set $\sigma \vcentcolon= \expecta_{\omega\sim\abstrd}\left[ s(\omega)\right]$ and $\tau \vcentcolon= \expecta_{\omega\sim\abstrd}\left[ t(\omega)\right]$.  The first line (next) uses these definitions and the second line is a simple re-organization:
\begin{align*}
    \frac{\sigma}{\tau} &= \frac{\expecta_{\omega\sim\abstrd}\left[ s(\omega)\right]}{\expecta_{\omega\sim\abstrd}\left[ t(\omega)\right]} \\
    0 &= \expecta_{\omega\sim\abstrd}\left[\tau\cdot s(\omega)-\sigma\cdot t(\omega) \right]
\end{align*}
\noindent 
Applying \Cref{fact:elementversusexpect} to the last line, there must exist $\omega_+$ for which $\tau\cdot s(\omega_+)-\sigma\cdot t(\omega_+)\geq 0$ which is equivalent to $\sfrac{s(\omega_+)}{t(\omega_+)}\geq \sfrac{\sigma}{\tau}$.
\end{proof}

\end{appendix}

\end{document}